\documentclass[conference]{IEEEtran}

\usepackage{mathtools}

\usepackage{floatflt}

\usepackage{color}
\usepackage{cite}

\usepackage{graphicx}
\usepackage{amsmath}
\usepackage{amsthm}
\usepackage{amssymb}

\usepackage{verbatim}

\usepackage{psfrag,array}

\usepackage{subfigure}

\usepackage{stfloats}

\usepackage{todonotes}

\usepackage{amsmath}
\usepackage{epstopdf}
\usepackage{algorithmic}

\newcommand{\p}[1]{\mathop{\mbox{\it p} } }

\renewcommand{\vec}[1]{\ensuremath{\boldsymbol{#1}}}
\newcommand{\be}{\begin{equation}}
\newcommand{\ee}{\end{equation}}
\newcommand{\ba}{\begin{array}}
\newcommand{\ea}{\end{array}}
\newcommand{\bea}{\begin{eqnarray}}
\newcommand{\eea}{\end{eqnarray}}
\newcommand{\bean}{\begin{eqnarray*}}
\newcommand{\eean}{\end{eqnarray*}}

\renewcommand{\Re}{\mathcal{R}}

\definecolor{white}{rgb}{1,1,1}

\newtheorem{theorem}{Theorem}

\newtheorem{proposition}{Proposition}
\newtheorem{property}{Property}

\newtheorem{corollary}{Corollary}

\begin{document}

\title{Cram\'{e}r-Rao Lower Bounds for Positioning with Large Intelligent Surfaces}
\author
{
Sha Hu, Fredrik Rusek, and Ove Edfors \\
Department of Electrical and Information Technology, \\
Lund University, Lund, Sweden \\
\{firstname.lastname\}@eit.lth.se
\vspace{-0mm}
}
\maketitle

\begin{abstract}
We consider the potential for positioning with a system where antenna arrays are deployed as a large intelligent surface (LIS). We derive Fisher-informations and Cram\'{e}r-Rao lower bounds (CRLB) in closed-form for terminals along the central perpendicular line (CPL) of the LIS for all three Cartesian dimensions. For terminals at positions other than the CPL, closed-form expressions for the Fisher-informations and CRLBs seem out of reach, and we alternatively provide approximations (in closed-form) which are shown to be very accurate. We also show that under mild conditions, the CRLBs in general decrease quadratically in the surface-area for both the $x$ and $y$ dimensions. For the $z$-dimension (distance from the LIS), the CRLB decreases linearly in the surface-area when terminals are along the CPL. However, when terminals move away from the CPL, the CRLB is dramatically increased and then also decreases quadratically in the surface-area. We also extensively discuss the impact of different deployments (centralized and distributed) of the LIS.
\end{abstract}

\section{Introduction}
Wireless communication has evolved from few and geographically distant base stations to more recent concepts involving a high density of access points, possibly with many antenna elements on each. A Large Intelligent Surface (LIS) is a newly proposed concept in wireless communication that is envisioned in \cite{HRE17, ewall}, where future man-made structures are electronically active with integrated electronics and wireless communication making the entire environment \lq\lq{}intelligent\rq\rq{} as depicted in Fig. 1. LIS can be seen as an extension of earlier research in several other fields. One strong relation is to the massive MIMO concept \cite{MM12, M10}, where large arrays of hundreds of antennas are used to achieve massive gains in spectral and energy efficiencies.

As LIS scales up beyond the traditional antenna array concept, it implies a clean break with the traditional access-point/base-station concept, as the entire environment is active in the communication. The natural limit of this evolution is that all LISs in an environment act as transmitting and receiving structures. LIS allows for an unprecedented focusing of energy in the three-dimensional space which enables, besides unprecedented data rates, wireless charging and remote sensing with extreme precision. This makes it possible to fulfill the most grand visions  in 5G communication and Internet of Things \cite{IOT} for providing connections to billions of devices. LIS seems to be first envisioned in the eWallpaper project at UC Berkeley \cite{ewall}. In \cite{HRE17}, we carry out a first analysis on information transfer capabilities of LIS, and show that, the number of signal space dimensions per $m^2$ deployed surface-area is $\pi/\lambda^2$, where $\lambda$ is the wavelength, and the capacity that can be harvested per $m^2$ surface-area is linear in the average transmit power, rather than logarithmic. 

\begin{figure}[t]
\begin{center}
\vspace*{-1mm}
\hspace*{-4mm}
\scalebox{0.6}{\includegraphics{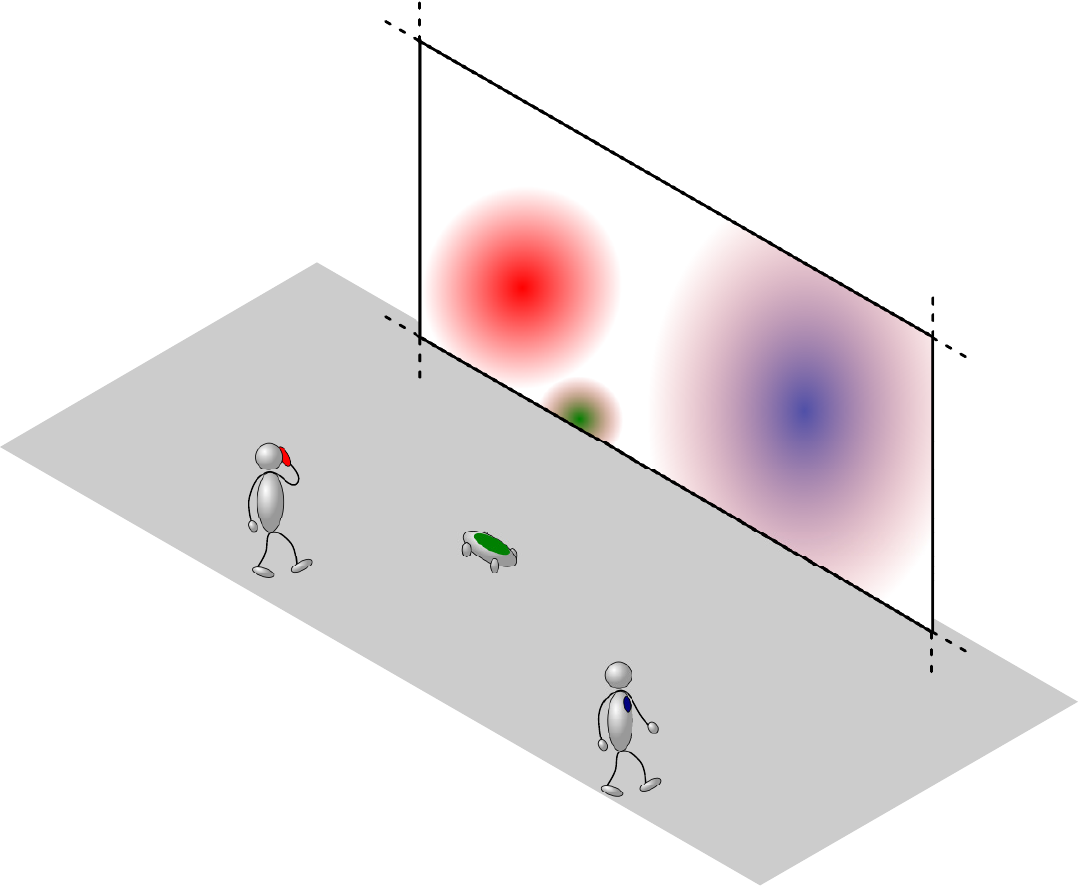}}
\vspace*{-3mm}
\caption{\label{fig1}Three users communicating with an LIS.}
\vspace*{-7mm}
\end{center}
\end{figure}

Following \cite{HRE17}, in this paper we take a first look at the potential of using LIS  for terminal positioning, where terminals are equipped with a single-antenna and located in the three-dimensional space in front of the LIS. For analytical tractability, we assume an ideal situation where no scatterers or reflections are present, yielding a perfect line-of-sight (LoS) propagation scenario and each terminal is assumed to radiate isotropically. Although we do not deal with more complicated geometries, our results are fundamental in the sense that positioning of objects in scattering environment \cite{ZB09, NV04} is commonly done in two steps: i) estimating the positions of a number of reflecting objects in the environment, and ii) backward computation of the position of the object of interest. Therefore, our results are instrumental for the understanding of the accuracy in the first step.

We derive the Cram\'{e}r-Rao lower bounds (CRLB) for terminals along the central perpendicular line (CPL) in closed-form. For a terminal that is not on the CPL, in order to analyze the properties of the CRLB, we use effective approximations of the Fisher-information and CRLB based on the results obtained with CPL. We approximate the Fisher-informations and CRLBs in such cases with closed-form expressions, which are shown to be very accurate under mild conditions. We also show that, the CRLB in general decreases quadratically in surface-area, except for terminals along the CPL where the CRLB for the $z$-dimension decreases linearly in the same. Meanwhile, the impact of wavelength is $\sim\!\lambda^2$. These scaling laws play in favor of LISs when compared to other positioning technologies e.g., optical systems. A LIS can compensate for its, comparatively, large wavelength by much larger aperture.

\section{Signal Model with LIS}
Expressed in Cartesian coordinates, we assume that the center of the LIS is located at coordinates $x\!=\!y\!=\!z\!=\!0$ and the terminal is located at positive $z$-coordinate. We assume an ideal situation with a perfect LoS propagation scenario where no scatterers or reflections are present, and each terminal radiates isotropically. Denoting the wavelength as $\lambda$, and assuming a narrow-band system and ideal free-space propagation from the terminal to that point, the received signal at the LIS at location $(x,\,y,\,0)$ radiated by a terminal at location $(x_0,\,y_0,\,z_0)$ is
\bea \label{md} \hat{s}_{x_0,\,y_0,\,z_0}(x,y)=s_{x_0,\,y_0,\,z_0}(x,y)+n(x,y), \eea
where $n(x,y)$ is modeled as independent complex Gaussian variable with a zero-mean and spectral density $N_0$, and the noiseless signal $s_{x_0,\,y_0,\,z_0}(x,y)$ is stated in Property 1.
\begin{property} 
The noiseless signal $s_{x_0,\,y_0,\,z_0}(x,y)$ can be described as
\bea \label{md1} s_{x_0,\,y_0,\,z_0}(x,y)=\frac{\sqrt{z_0}}{2\sqrt{\pi}\eta^{3/4}}\exp\!\left(\!-\frac{2\pi j\sqrt{\eta}}{\lambda}\right)\!,   \eea
where
\bea  \label{eta} \eta\!=\!z_0^2\!+\!(y\!-\!y_0)^2\!+\!(x\!-\!x_0)^2.\eea 
\end{property} 
\begin{proof}
The radiating model of transmitting signal at location $(x_0,y_0,z_0)$ to the LIS is depicted in Fig. \ref{figsm}. The noiseless signal received by the LIS at location $(x,\,y,\,0)$ and time epoch $t$ reads,
\bea \label{nbs1} s_{x_0,\,y_0,\,z_0}(x,y)=\sqrt{P_L\cos\phi(x,y)}s(t) \exp\!\left(\!-2\pi j f_c t\right), \eea
where $P_L$ denotes the path-loss, $\phi(x,y)$ is angle-of-arrivals (AoA) of the transmitted signal $s(t)$ at $(x,y,0)$ , and $f_c$ is the carrier-frequency.  The transmit-time from the terminal to the $(x,y,0)$ equals $t_0=\frac{\sqrt{\eta}}{c}$, where $c$ is the speed of light. Since we are considering a narrow-band system, the transmit signal $s(t)$ can be assumed to be constant over time-interval $[0, t_0]$, therefore, we assume $s(t)\!=\!1$ and remove it from (\ref{nbs1}). Further, as the free-space path-loss is $P_L\!=\!\frac{1}{4\pi\eta}$ and $\cos\phi(x,y)\!=\!\frac{z_0}{\sqrt{\eta}}$, inserting them back into (\ref{nbs1}) yields (\ref{md1}).
\end{proof}

\begin{figure}[t]
\begin{center}
\vspace*{-2mm}
\hspace*{-2mm}
\scalebox{0.82}{\includegraphics{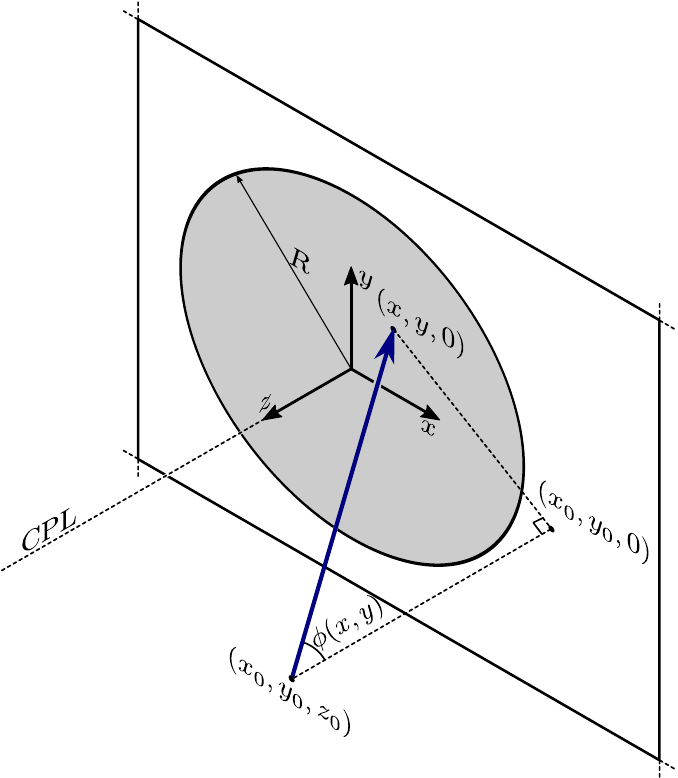}}
\vspace*{-2mm}
\caption{\label{figsm}The radiating model of transmitting signal to the LIS. We integrate the received signal of each point-element over the whole area spanned by the LIS. Therefore, for each point-element on the LIS, the Fraunhofer distance \cite{F10} is infinitely small and the received signal model (\ref{md1}) holds for both near-filed and far-field scenarios.}
\vspace*{-6mm}
\end{center}
\end{figure}

Denote the first-order derivatives with respect to variables $x_0$, $y_0$ and $z_0$ as $\Delta s_1$, $\Delta s_2$, and $\Delta s_3$, respectively. We have
{\setlength\arraycolsep{0pt}  \bea \label{dev1} \Delta s_1&=&\frac{\sqrt{z_0}\left(x -x_0\right)}{2\sqrt{\pi}}\!\left(\!
\frac{3}{2}\eta^{-\frac{7}{4}} \!+\! \frac{2\pi j }{\lambda}{\eta}^{-\frac{5}{4}}\!\right)\!\exp\left(\!-\frac{2\pi j\sqrt{\eta}}{\lambda} \right)\!, \notag  \\
\label{dev2} \Delta s_2&=&\frac{\sqrt{z_0}\left(y -y_0\right)}{2\sqrt{\pi}}\left(
\!\frac{3}{2}\eta^{-\frac{7}{4}} \!+\! \frac{2\pi j }{\lambda}{\eta}^{-\frac{5}{4}}\!\right)\!\exp\left(\!-\frac{2\pi j\sqrt{\eta}}{\lambda}\right)\!,  \notag  \\
\label{dev3} \Delta s_3&=&\frac{z_0^{\frac{3}{2}}}{2\sqrt{\pi}}\!\left(
\frac{1}{2z_0^2}\eta^{-\frac{3}{4}} \!- \!\frac{3}{2}\eta^{-\frac{7}{4}} \!- \!\frac{2\pi j }{\lambda}{\eta}^{-\frac{5}{4}}\!\right)\!\exp\left(\!-\frac{2\pi j\sqrt{\eta}}{\lambda} \right)\!.  \notag \\ \eea}
\hspace{-1.4mm}Then, the elements of the Fisher-information matrix are given by the following double integral
\bea \label{Fisherij} I_{ik}=\frac{2}{N_0}\iint_{x, y}\Re\!\left\{\Delta s_k \left(\Delta s_i\right)^{\ast}\right\}\mathrm{d}x\mathrm{d}y, \eea
where \lq{}$\Re\{\}$\rq{} takes the real part and the integral is taken over the area of the LIS, which we assume to have a disk-shape of radius $R$. Note that as CRLB scales down linearly in signal to noise ratio (SNR), we set $N_0=2$ throughout the paper to eliminate the scaling factor in (\ref{Fisherij}).

\section{CRLBs of Terminals along the CPL}
In this section we analyze the CRLBs for terminals along the CPL, i.e., with coordinates $(0, 0, z_0)$. A nice property of the CPL is that, the CRLBs for all dimensions are in closed-form, and can be used to approximate the CRLBs for terminals at non-CPL positions. We denote the Fisher-information and CRLB for terminals with coordinates $(x_0, y_0,z_0)$ and a LIS with radius $R$ as $I_i([x_0, y_0,z_0], R)$ and $C_i([x_0, y_0,z_0], R)$, where the suffix $i=x, y, z$ represents the $x$, $y$, $z$ dimension, respectively. When $i$ contains multiple variables, it means that all these dimensions contained in $i$ have the same value. For instance, $I_{x, y}([x_0, y_0,z_0], R)$ denotes the Fisher-information for both $x$ and $y$ dimensions whenever they are equal.

\subsection{CRLBs for Three Dimensions}
\begin{theorem}
The Fisher-information matrix \vec{I} for terminals with coordinates $(0,0,z_0)$ is diagonal, and the elements are
{\setlength\arraycolsep{2pt} \bea  \label{Ixy} I_{11}&=&I_{22}=I_{x, y}([0,0,z_0], R)=-\frac{B_1}{A},  \\
 \label{Iz}  I_{33}&=&I_z([0,\,0,\, z_0], R)=-\frac{2B_2}{A}, \eea}
\hspace{-1.4mm}where $A$, $B_1$, and $B_2$ are defined in (\ref{A1})-(\ref{B2}).
\end{theorem}
\begin{proof}
Theorem 1 is proved by directly solving the integrations in (\ref{Fisherij}), see Appendix A.
\end{proof}
\begin{figure*}[b]
\vspace{-4mm}
\hrulefill
{\setlength\arraycolsep{2pt} \bea \label{A1} A&=&240\lambda^2z_0^2{\left(R^2 + z_0^2\right)}^{\frac{5}{2}},   \\
B_1&=&160{\pi}^2 z_0^7 + 400 {\pi}^2R^2 z_0^5 + 240{\pi}^2 z_0^3 R^4  - 18\lambda^2 {\left(R^2 + z_0^2\right)}^{\frac{5}{2}}- 160 {\pi}^2 z_0^2  {\left(R^2 + z_0^2\right)}^{\frac{5}{2}} + 18\lambda^2z_0^5 + 45\lambda^2 R^2 z_0^3,  \\
 \label{B2} B_2&=&12\lambda^2z_0^5 - 12 \lambda^2 {\left(R^2 + z_0^2\right)}^{\frac{5}{2}} + 80{\pi}^2z_0^7  - 80{\pi}^2z_0^2 {\left(R^2 + z_0^2\right)}^{\frac{5}{2}}+ 15z_0\lambda^2  R^4 + 80{\pi}^2R^2z_0^5.\eea}
\vspace{-10mm}
\end{figure*}

According to Theorem 1, the CRLB can be computed as
\bea C_i([0,\,0,\, z_0], R)=I_{i}^{-1}([0,\,0,\, z_0], R), \; i=x, y, z, \eea
and the following conclusions can be derived. 

Firstly, when $z_0$ is close to 0, that is, the terminal is close\footnote{But $z_0=0$ is a singularity point and the CRLBs are $\infty$ as no signal is received by the LIS.} to the LIS, the CRLBs $C_i([0,\,0,\, z_0], R)$ are 0, while under the case $R\ll z_0$, the CRLBs are $\infty$. These observations are consistent with the nature of the problem at hand.

Secondly, in order to get a direct view of the CRLBs in relation to the surface-area, we assume $\lambda\ll z_0$ (which in general holds as $\lambda$ is the wavelength). Defining
\bea\tau=R/z_0,\eea
then the CRLBs can be simplified as
{\setlength\arraycolsep{2pt} \bea \label{Cxy}C_{x,y}([0,\,0,\, z_0], R)
&\approx&\frac{3\lambda^2}{2\pi^2}f_1(\tau), \\
 \label{Cz} C_z([0,\,0,\, z_0], R)
&\approx&\frac{3\lambda^2}{2\pi^2}f_2(\tau), \eea}
\hspace{-1.4mm}where
\bea f_1(\tau)&=&\frac{\left(1+ \tau^2\right)^{\frac{5}{2}}}{\left(1+ \tau^2\right)^{\frac{5}{2}}-1\, - 2.5\, \tau^{2}- 1.5\, \tau^{4}},  \\
\label{f2} f_2(\tau)&=&\frac{\left(1+ \tau^2\right)^{\frac{5}{2}}}{\left(1+ \tau^2\right)^{\frac{5}{2}} -1 - \tau^{2}}, \eea
respectively. As can been seen, the CRLBs for all dimensions are uniquely decided by $\tau$. Hence, when $z_0$ is increased by a factor of $\alpha$, the radius $R$ of the LIS also has to increase by the same factor in order to have the same CRLBs. Another interesting but somewhat intuitive fact is that the CRLBs for $x$ and $y$ dimensions are higher than that for $z$-dimension.

Thirdly, under the case $R\!\gg \!z_0$, i.e., $\tau\!\to\!\infty$, the asymptotic CRLBs in (\ref{Cxy}) and (\ref{Cz}) are identical and equal
\bea \label{limtCRLB} \lim_{R\to \infty} C_{x, y, z}([0,\,0,\, z_0], R)=\frac{3\lambda^2}{2\pi^2},\eea
which depends only on the wavelength $\lambda$ and represents a fundamental lower limit to positioning precision.

Lastly, in reality, the most likely case is $R\ll z_0$, i.e., $\tau\to 0$. Then we have the approximations 
{\setlength\arraycolsep{2pt}\bea f_1(\tau)&=&\frac{8}{3}\tau^{-4}+o\!\left(\tau^{-4}\right),   \\ 
f_2(\tau)&=&\frac{2}{3}\tau^{-2}+o\!\left(\tau^{-2}\right),\eea}
\hspace{-1.4mm}for sufficient small $\tau$, and the CRLBs can be approximated as
{\setlength\arraycolsep{2pt}\bea \label{area1} C_{x,y}([0,\,0,\, z_0], R)&\approx&\frac{4 \lambda^2}{\pi^2 \tau^4},   \\
\label{area2} C_z([0,\,0,\, z_0], R)&\approx&\frac{\lambda^2}{\pi^2\tau^2}. \eea}
\hspace{-2.1mm}This shows that along the CPL, the CRLBs for $x$ and $y$ dimensions decrease quadratically with the surface-area ($\tau^2$ is proportional to area), while the CRLB for the $z$-dimension decreases linearly in that. This is so, since the CRLB for the $z$-dimension is much lower than that for the other two dimensions, but as $\tau$ increases, the limits of the CRLB for all three dimensions are identical as in (\ref{limtCRLB}). Therefore, the CRLBs for $x$ and $y$ dimensions must decrease faster than that for the $z$-dimension as surface-area increases.

\begin{figure*}[b]
\begin{center}
\vspace*{-76mm}
\hspace*{-2mm}
\scalebox{0.55}{\includegraphics{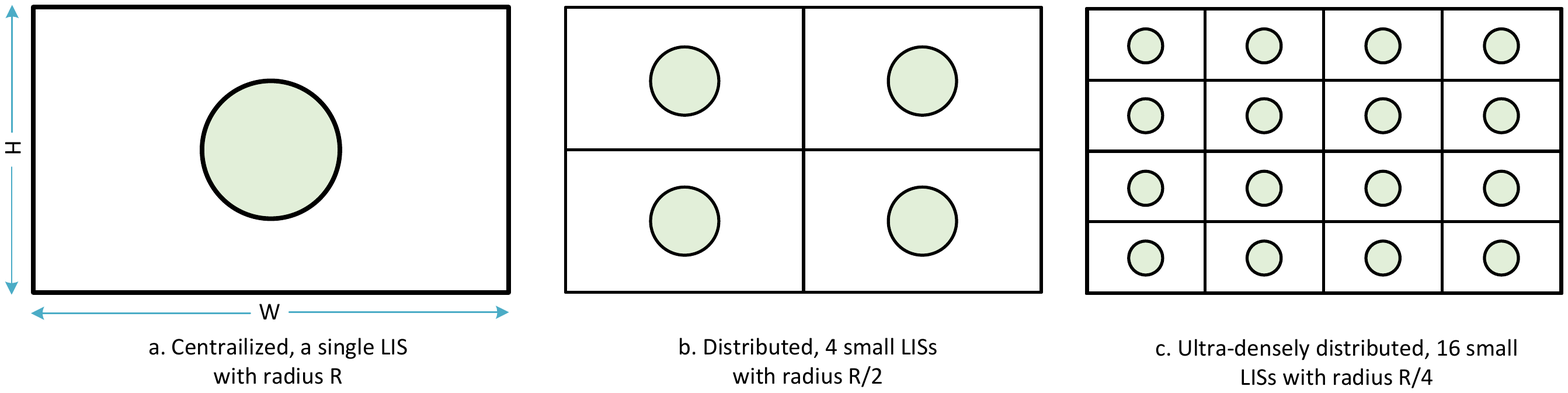}}
\vspace*{-5mm}
\caption{\label{fig4}Different deployments of the LIS on a surface with width $W$ and height $H$. The total surface-area is the same for different deployments. Note that, splitting a LIS into small LISs has a cost of communication channels among different small LISs for cooperation.
}
\vspace*{-4mm}
\end{center}
\end{figure*}

\section{CRLBs of Terminals not on the CPL}
Instead of considering terminals along the CPL, in this section we consider terminals with arbitrary coordinates $(x_0, y_0, z_0)$. When $x_0, y_0\!\neq\!0$, closed-form expressions of the CRLB seem out of reach due to the complicated integrations in (\ref{Fisherij}). Therefore, we seek approximations, tight enough so that insights can still be drawn, of the CRLBs in this case. Using the closed-form expressions of Fisher-information for terminals along the CPL in Sec. III, the CRLBs for the general case can be well-approximated as elaborated in detail next.

\subsection{CRLB Approximations for Terminals with Arbitrary Coordinates $(x_0, y_0, z_0)$}
We first introduce two mild conditions,
{\setlength\arraycolsep{2pt}  \bea \label{cond1}  \lambda&\ll& \frac{z_0^2}{\sqrt{z_0^2+x_0^2+y_0^2+R^2}}, \\
  \label{cond2} 2R&\ll&\frac{z_0^2}{\sqrt{x_0^2+y_0^2}}+\sqrt{x_0^2+y_0^2}.\eea}
\hspace{-1.4mm}As for the cases of interest $R$ is relatively small compared to $z_0$, and $\lambda$ is much smaller than $z_0$, these two conditions are usually satisfied. Letting
\bea z_1=\sqrt{x_0^2 + y_0^2 + z_0^2},\eea
our approximations for Fisher-information and CRLBs are stated in Proposition 1.

\begin{proposition}
Under the conditions (\ref{cond1})-(\ref{cond2}), the Fisher-information matrix for terminals with coordinates $(x_0, y_0, z_0)$ can be approximated as
{\setlength\arraycolsep{5pt}  \bea  \label{Imat} \vec{I}\approx\left[\!\begin{array}{ccc} \alpha+\frac{\beta\, x_0^2}{z_0^2}  & \frac{\beta\, x_0\, y_0}{z_0^2} & \frac{\beta\, x_0}{z_0}\\ \frac{\beta\, x_0\, y_0}{z_0^2} & \alpha+\frac{\beta\, y_0^2}{z_0^2}  & \frac{\beta\, y_0}{z_0}\\ \frac{\beta\, x_0}{z_0} & \frac{\beta\, y_0}{z_0} &\beta \end{array}\right]\!, \eea}
\hspace{-1.4mm}where $\alpha$ and $\beta$ equal
{\setlength\arraycolsep{2pt} \bea \label{approx1} \alpha&\approx& \frac{z_0}{z_1}I_{x,y}([0, 0, z_1], R),\\
 \label{approx2} \beta&\approx& \left(\frac{z_0}{z_1}\right)^3I_z([0, 0, z_1], R), \eea}
\hspace{-1.4mm}and $I_{x,y}([0, 0, z_1], R)$ and $I_z([0, 0, z_1], R)$ are the Fisher-informations for $x, y$ and $z$ dimensions for terminals with coordinates $(0, 0, z_1)$ as given in Theorem 1. Then, the CRLB matrix equals
{\setlength\arraycolsep{5pt}  \bea  \label{Cmat} \vec{C}=\vec{I}^{-1}\approx\left[\!\begin{array}{ccc} \frac{1}{\alpha} & 0 & -\frac{x_0}{\alpha\, z_0}\\ 0 & \frac{1}{\alpha} & -\frac{y_{0}}{\alpha\, z_0}\\ -\frac{x_{0}}{a\, z_0} & -\frac{y_{0}}{\alpha\, z_0} &\frac{1}{ \beta} + \frac{ x_0^2+y_0^2}{\alpha z_0^2} \end{array}
\right]\!. \eea}
\end{proposition}
\begin{proof} 
The derivations underlying the approximations are based on (\ref{cond1})-(\ref{cond2}) and utilizing Theorem 1 to approximate the integrations in (\ref{Fisherij}), but are omitted here due to page-limit.
\end{proof}

From Proposition 1, the Fisher-informations and CRLBs are approximated in closed-form. Specially, when $x_0\!=\!y_0\!=\!0$, that is, terminals are along the CPL line,  the approximations in (\ref{Imat})-(\ref{Cmat}) are equalities. Further, we have the below corollary.
\begin{corollary}
Under the conditions (\ref{cond1})-(\ref{cond2}), the CRLBs for the $x$ and $y$ dimensions are approximately equal, and depend on $(x_0, y_0, z_0)$ through $z_0$ and the radius $\sqrt{x_0^2+y_0^2}$.
\end{corollary}

To simplify the analysis, we assume $R\!\ll\!z_0$. Then using (\ref{area1})-(\ref{area2}) and Proposition 1, we have the approximated CRLBs stated in Proposition 2.

\begin{proposition}
Under the case $R\!\ll \!z_0$ and with conditions (\ref{cond1})-(\ref{cond2}), the CRLBs for terminals with coordinates $(x_0, y_0, z_0)$ can be approximated as
{\setlength\arraycolsep{2pt}\bea  \label{app1} C_{x,y}&\approx&\frac{4\lambda^2z_1^5}{\pi^2z_0R^4},   \\
  \label{app2} C_{z}&\approx&\frac{\lambda^2z_0^2}{\pi^2R^2}+\frac{4\lambda^2(x_0^2+y_0^2)z_1^5}{\pi^2z_0^3R^4}. \eea}
\end{proposition}

Compared to the case that terminals are along the CPL, with a relatively small $R$, the CRLB for the $z$-dimension is dramatically increased when terminals move away from the CPL, that is, $x_0^2\!+\!y_0^2\!\neq\!0$. Further, when $\sqrt{x_0^2\!+\!y_0^2}\!>\!z_0$, the CRLB for the $z$-dimension becomes even larger than that for the $x$ and $y$ dimensions. Most importantly, when terminals are away from the CPL, the CRLBs decrease quadratically in the surface-area of the LIS for all three dimensions. Therefore, the LIS can provide substantial gains over the massive MIMO for terminal positioning, as the number of antenna-elements deployed in an LIS is increased by a factor of 10$\sim$100 over a traditional massive MIMO deployment.

\subsection{CRLB for AoA and Radius}
Instead of estimating the coordinates $(x_0, y_0, z_0)$, in some cases it is of more interest to estimate the angles-of-arrival (AoA) $\phi, \psi$ and the corresponding radius $\kappa\!=\!\sqrt{x_0^2+y_0^2+z_0^2}$, in which case, we have spherical coordinates as
{\setlength\arraycolsep{2pt} \bea \label{spcart} 
x_0&=&\kappa\sin\phi\cos\psi,\notag \\
y_0&=&\kappa\sin\phi\sin\psi, \notag \\
z_0&=&\kappa\cos\phi. \eea}
\hspace{-1.1mm}Using variable substitution formula \cite{K93} for CRLB computation, the CRLB matrix for estimating $(\kappa,\phi, \psi)$ can be directly calculated based on Proposition 1 and 2.

\section{Deployments of the LIS}
In this section we consider different deployments of the LIS on a surface with size $W\!\times\! H$ where $W, H$ are the width and height, respectively. In particular, we consider the centralized-deployment (a) and the distributed-deployments (b) and (c) as depicted in Fig. \ref{fig4}. For simplicity, we assume $R,\,\lambda\!\ll\! z_0$ and consider the CRLB for a terminal on the CPL with coordinates $(0, 0, z_0)$, that is, positioning a terminal on the far-field.

With the centralized deployment (a), the CRLBs are given in (\ref{area1}) and (\ref{area2}). With a distributed deployment (b), the LIS is split into four small LISs centered at $(\pm W/4, \pm H/4)$ with equal radius $R/2$. Using Proposition 1 and 2, and the symmetry between the centers of the LISs and the terminal-positions, the sum of the Fisher-information matrices corresponding to the four small LISs is diagonal, and the Fisher-information for the $x, y$ and $z$ dimensions can be shown to equal
{\setlength\arraycolsep{2pt} \bea I_{x,y}&\approx&\frac{\pi^2z_0R^4}{16\lambda^2(z_0^2+D^2)^{5/2}}+\frac{\pi^2D^2z_0R^2}{2\lambda^2(z_0^2+D^2)^{5/2}},  \\
I_z&\approx&\frac{\pi^2R^2z_0^3}{\lambda^2(z_0^2+D^2)^{5/2}}.  \eea}
\hspace{-1.4mm}where $D\!=\!\frac{\sqrt{W^2+H^2}}{4}$. Assuming $D\!\ll\! z_0$, it holds that
 {\setlength\arraycolsep{2pt}    \bea  \label{Ixy4} I_{x,y}&\approx&\frac{\pi^2R^4}{4\lambda^2z_0^4}\left(\frac{1}{4}+\frac{2D^2}{R^2}\right),  \\
 \label{Iz4} I_{z}&\approx&\frac{\pi^2R^2}{\lambda^2z_0^2}. \eea}
\hspace{-1.4mm}As can be seen, compared to the centralized case, the CRLBs for all dimensions with the distributed deployment decrease linearly in the total surface-area for a relatively small $R$. 

Further, comparing (\ref{area1}) to (\ref{Ixy4}), we obtain the insight that the CRLBs for $x$ and $y$ dimensions with the distributed deployment (b) is lower than that with the centralized deployment (a) only if $D\!>\!\sqrt{\frac{3}{8}}R$, or equivalently, 
 \bea \label{thresh2} R< \sqrt{\frac{W^2+H^2}{6}}. \eea

That is to say, in the far-field with the distributed deployment (b), the CRLB for the $x$ and $y$ dimensions is improved for a terminal with a distance to the CPL larger than $\sqrt{6}R$. Otherwise, the centralized deployment (a) provides lower CRLB for $x, y$ dimensions than that for the distributed deployment. However, the CRLB for the $z$-dimension remains the same. Following the same principle, under the condition (\ref{thresh2}), one can continue to split the LISs into more small pieces and obtain an ultra-densely distributed deployment such as in (c). In general, the positioning performance is further improved with deployment (c) as the projection to the LIS-plane of each terminal-position is centered by a number of small LISs.

\section{Numerical  Results}
\subsection{Exact-CRLB Evaluations}
We first evaluate the CRLBs for terminals along and away from the CPL. As only the radius $\sqrt{x_0^2\!+\!y_0^2}$ matters as shown in Corollary 1, we test with offsets only in the $x$-dimension. In Fig. \ref{fig5} and \ref{fig7}, we test with $R\!=\!1$, $\lambda\!=\!0.1$, $y_0\!=\!0$, $x_0\!=\!2$, 4, 8, and $z_0\!=\!4$, 6, respectively. Some interesting results can be observed. Firstly, as shown in Fig. \ref{fig5}, when $\tau$ is small the CRLB for the $x$ and $y$ dimensions decrease quadratically in surface-area, while as shown in Fig. \ref{fig7}, the CRLB for the $z$ -dimension decreases linearly in that. This is well aligned with (\ref{area1}) and (\ref{area2}). Secondly, the CRLB for the $z$-dimension increases dramatically when the terminal moves away from the CPL. Further, as long as $x_0 \!\neq \!0$, the CRLB for the $z$-dimension also decreases quadratically in the surface-area. These phenomenons are well predicted by Proposition 1 and 2. Lastly, it can been seen that, as $R\to\infty$, the CRLBs converge to the limit  $\frac{3\lambda^2}{2\pi^2}\!=\! 1.5\!\times\!10^{-3}$ for all dimensions as shown in (\ref{limtCRLB}).

\subsection{CRLB Approximation Accuracies}

Next we evaluate the CRLB approximations accuracies for terminals at non-CPL positions. We compare the numerical integration results of the CRLB (with absolute error $10^{-10}$ and relatively error $10^{-6}$ using the Matlab built-in function \lq{}integral\rq{}) and the approximations using (\ref{app1})-(\ref{app2}) in Proposition 2. We test with $R\!=\!0.5$, $\lambda\!=\!0.1$, and $z_0\!=\!8$, and set $x_0\!=\!y_0$ in the range from 1 to 8. The CRLBs and the normalized approximation errors are shown in Fig. \ref{fig9}, where the normalized errors are computed as the normalized CRLB differences between the numerical integrations and the approximations. As can be seen, the approximations of CRLB given by Proposition 2 perform well, with errors less than 0.5\% for the $x, y$ dimensions and close to 1\% for the $z$-dimension. The errors for the $z$-dimension are slightly higher than those for the $x, y$ dimensions is because the approximations depend both on estimations of $\alpha$ and $\beta$, rather than only $\alpha$ as the latter case, as shown Proposition 1.

\subsection{CRLB with Different Deployments}
At last, we evaluate the CRLB with centralized and distributed deployments discussed in Sec. V. We set $W\!=\!H\!=\!4$ and $z_0\!=\!8$. All curves are obtained with numerical integrations and we compare the CRLBs for the deployments depicted in Fig. 2, that is, a single LIS, 4 small LISs, and 16 smaller LISs, with the same total surface-area. As shown in Fig. 7, when (\ref{thresh2}) is fulfilled, i.e., $R\!\leq\!\sqrt{\frac{W^2+H^2}{6}}\!=\!2.31$, the distributed deployments with 4 and 16 small LISs render lower CRLBs than the centralized LIS for the $x$ and $y$ dimensions, while the CRLBs for the $z$-dimension are almost the same. When $R$ increases beyond the limit, the distributed deployments become worse for the $x$ and $y$ dimensions, although the $z$-dimension is slightly better.  As $R$ increases, different deployments converge to each other as expected. In addition, further splitting the 4 small LISs into 16 smaller LISs only provides marginal gains at a cost of more communication channels are needed for different small LISs to cooperating with each other.

\begin{figure}[t]
\begin{center}
\vspace*{-4mm}
\hspace*{-5mm}
\scalebox{0.335}{\includegraphics{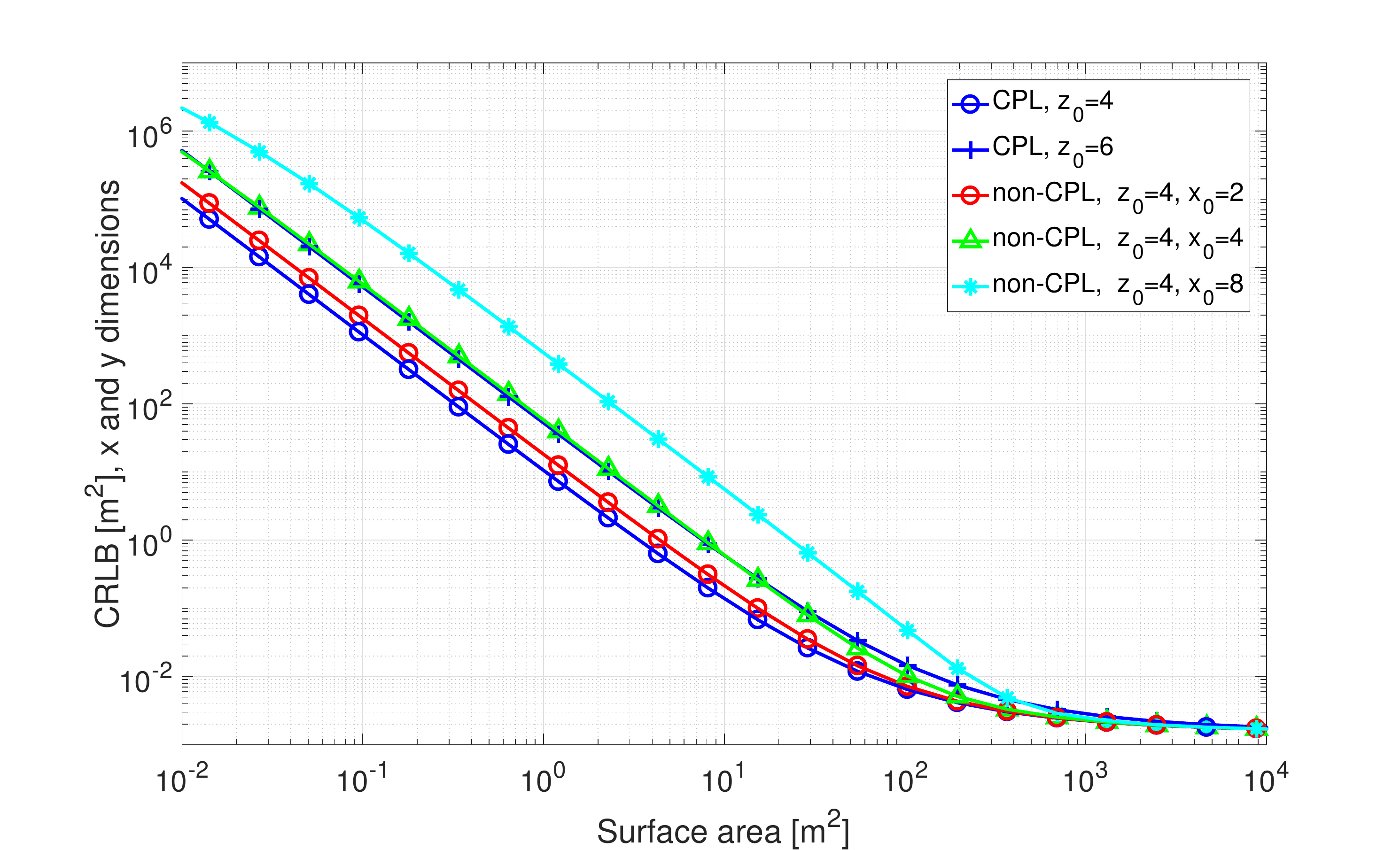}}
\vspace*{-7mm}
\caption{\label{fig5}CRLB for $x$ and $y$ dimensions, and the CRLBs for $y$-dimension are almost overlapped with those for $x$-dimension.}
\vspace*{-5mm}
\end{center}
\end{figure}

\begin{figure}
\begin{center}
\vspace*{-2mm}
\hspace*{-5mm}
\scalebox{0.335}{\includegraphics{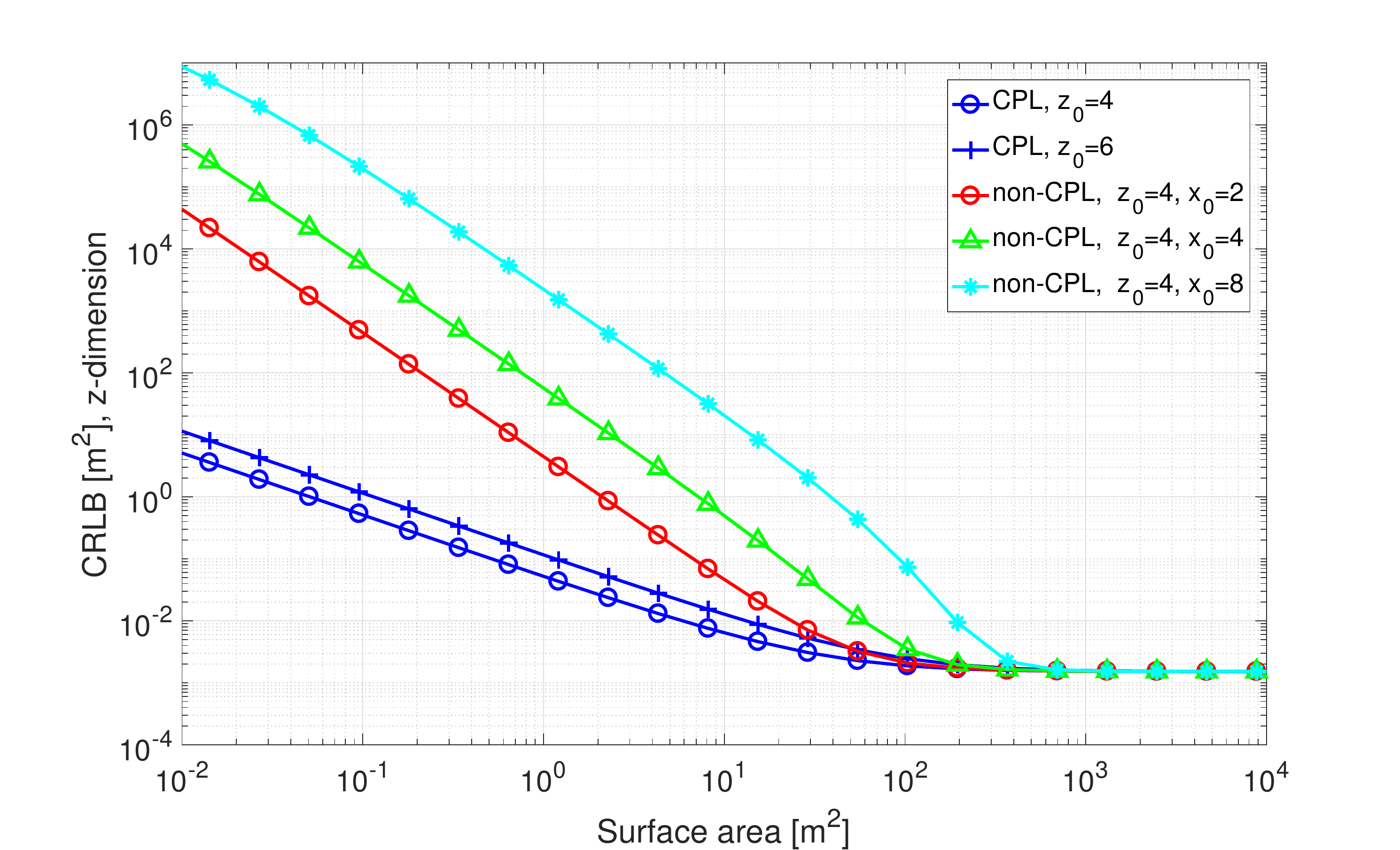}}
\vspace*{-7mm}
\caption{\label{fig7}CRLBs for the $z$-dimension for the cases evaluated in Fig. \ref{fig5}.}
\vspace*{-6mm}
\end{center}
\end{figure}

\begin{figure}[t]
\begin{center}
\vspace*{-3mm}
\hspace*{-7mm}
\scalebox{0.335}{\includegraphics{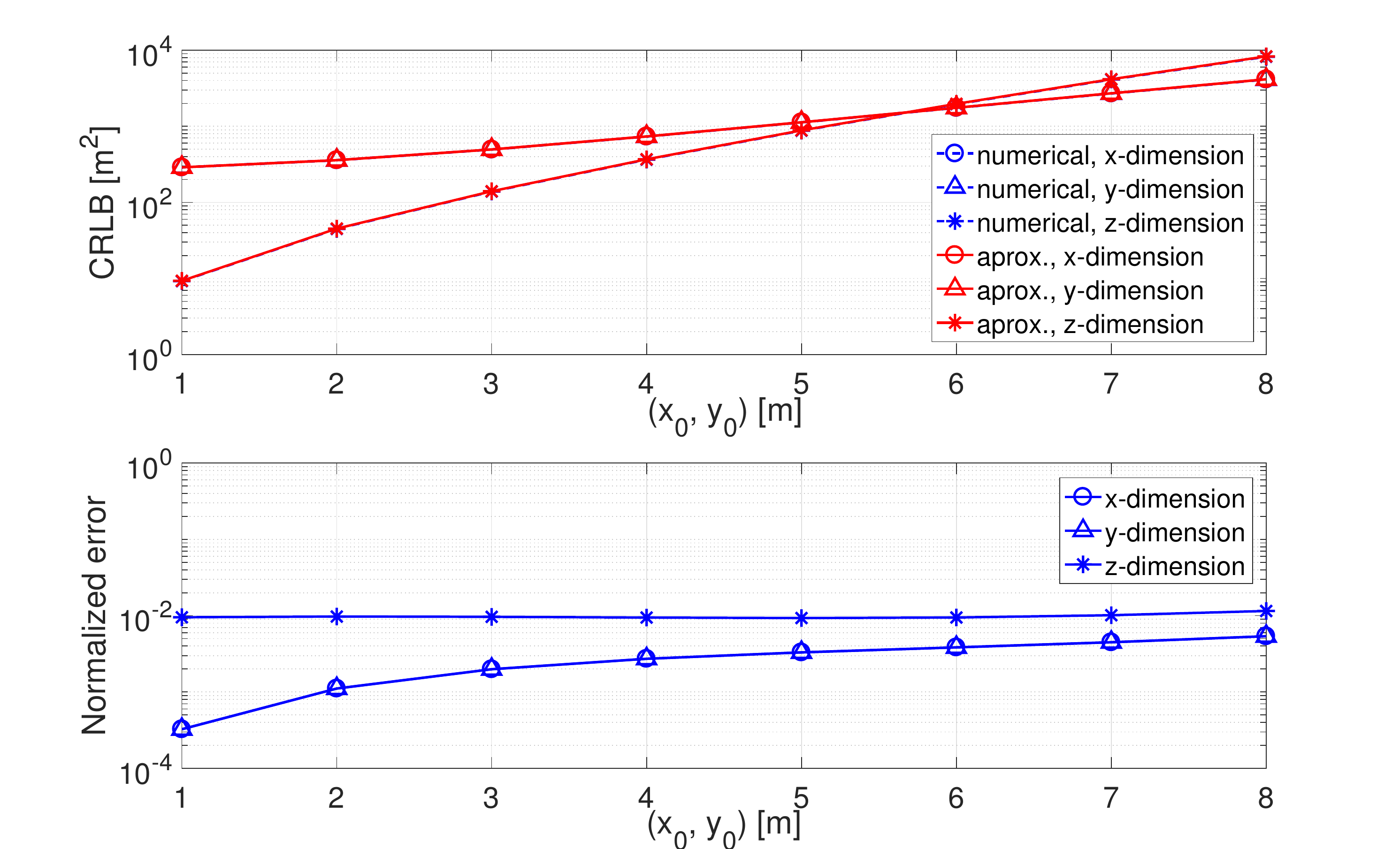}}
\vspace*{-7mm}
\caption{\label{fig9}CRLBs evaluated with numerical integrations and the approximations using Proposition 1, and the normalized approximation errors for different dimensions. The curves are almost on top of each other for the CRLB, while the normalized approximation errors for the $x$ and $y$ dimensions are the same.}
\vspace*{-4mm}
\end{center}
\end{figure}

\begin{figure}
\begin{center}
\vspace*{-3mm}
\hspace*{-7mm}
\scalebox{0.335}{\includegraphics{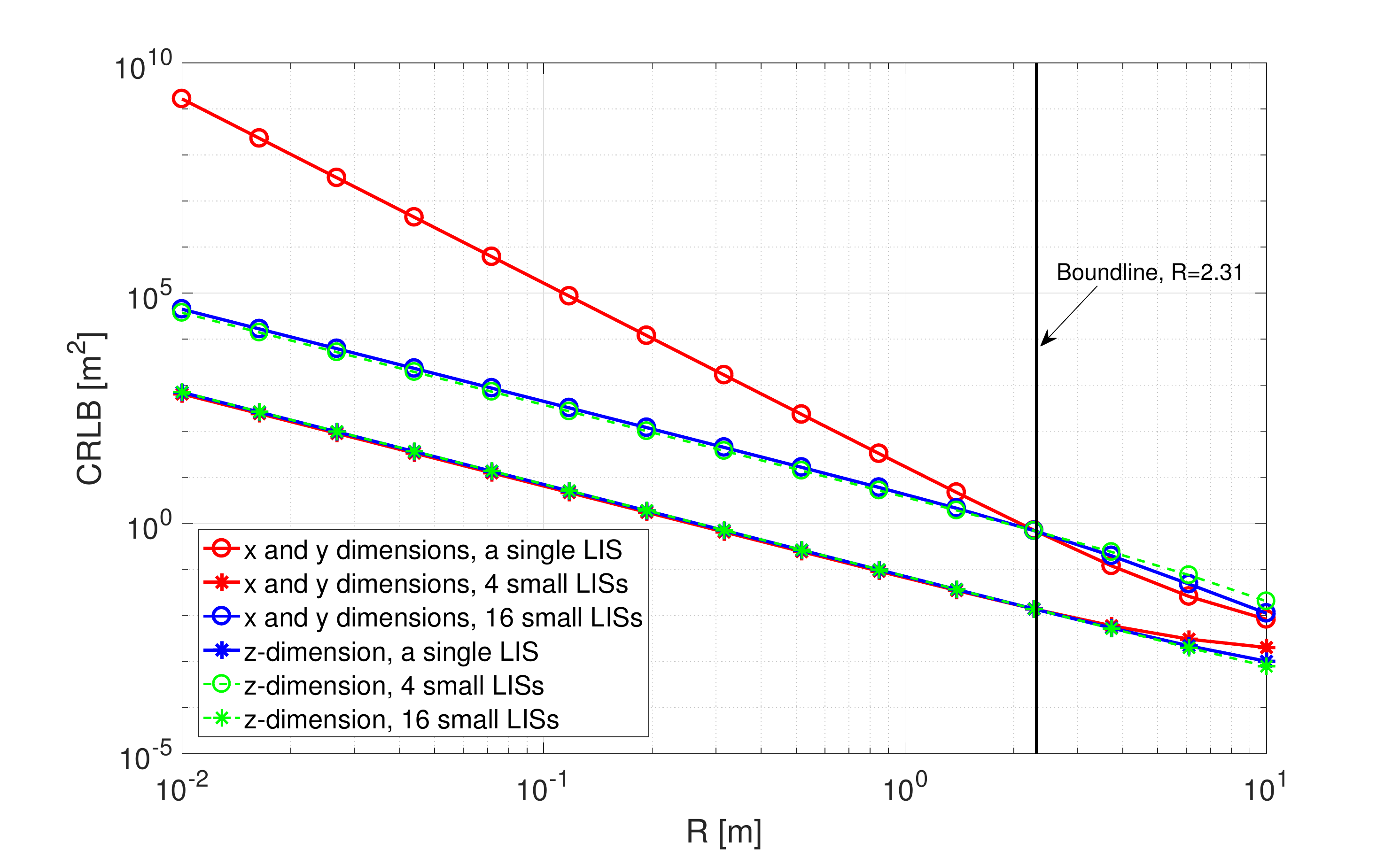}}
\vspace*{-7mm}
\caption{\label{fig10}The CRLBs evaluated with different deployments of the LIS.}
\vspace*{-6mm}
\end{center}
\end{figure}

\section{Summary}
In this paper, we have derived the Fisher-information and Cram\'{e}r-Rao lower bound (CRLB) for terminal positioning with large intelligent surfaces (LIS). For terminals along the central perpendicular line (CPL), the CRLBs are derived in closed-form for all Cartesian dimensions. For other positions we have alternatively provided approximations in closed-forms to compute the Fisher-informations and CRLBs which are shown to be accurate. We have shown that, under mild conditions the CRLBs in general decrease quadratically in the surface-area of the deployed LIS for all dimensions. Moreover, we compare centralized and distributed deployments of the LIS, and show that the distributed deployments have the potential to lower the CRLBs as long as the surface-area is less than a certain limit.

\section*{Appendix A}
Firstly, we define two functions $g_1(n)$ and $g_2(n)$ as
{\setlength\arraycolsep{0pt}\bea  \label{g1}  g_1(n)&=&\iint_{x^2+y^2\leq R^2}x^2\eta^{-\frac{n}{2}}\mathrm{d}x\mathrm{d}y 
=\iint_{x^2+y^2\leq R^2}y^2\eta^{-\frac{n}{2}}\mathrm{d}x\mathrm{d}y,  \notag \\
\label{g2} g_2(n)&=&\iint_{x^2+y^2\leq R^2}\eta^{-\frac{n}{2}}\mathrm{d}x\mathrm{d}y. \notag \eea}
\hspace{-1.4mm}In general, closed-form expressions of $g_1(n)$ and $g_2(n)$ are out of reach, except for the case that $x_0\!=\!y_0\!=\!0$, i.e., the terminal is on the CPL, in which case, $g_1(n)$ and $g_2(n)$ are in closed-form and it holds that
{\setlength\arraycolsep{0pt} \bea \label{g1n} g_1(n)&=&\frac{\pi\!\left(2z_0^{4 -n}\! -\! {\left(R^2 + z_0^2\right)}^{1 - \frac{n}{2}}\! \left(n R^2 \!-\!2 R^2 + 2z_0^2\right)\!\right)}{n^2 - 6n + 8},  \qquad \\
\label{g2n} g_2(n)&=&\frac{z_0^{2 -n}\! -\! {\left(R^2 \!+\! z_0^2\right)}^{1 \!-\! \frac{n}{2}}}{n - 2}.  \eea}
\hspace{-1.4mm}For a terminal on the CPL, as $x_0\!=\!y_0\!=\!0$, the first-order derivatives with respect to $x$ and $y$ are equals to
{\setlength\arraycolsep{0pt}  \bea \label{dev11} \Delta s_1&=&\frac{\sqrt{z_0}x}{2\sqrt{\pi}}\!\left(\!
\frac{3}{2}\eta^{-\frac{7}{4}} \!+\! \frac{2\pi j }{\lambda}{\eta}^{-\frac{5}{4}}\!\right)\!\exp\left(\!-\frac{2\pi j\sqrt{\eta}}{\lambda} \right)\!,   \\
\label{dev22} \Delta s_2&=&\frac{\sqrt{z_0}y}{2\sqrt{\pi}}\left(
\!\frac{3}{2}\eta^{-\frac{7}{4}} \!+\! \frac{2\pi j }{\lambda}{\eta}^{-\frac{5}{4}}\!\right)\!\exp\left(\!-\frac{2\pi j\sqrt{\eta}}{\lambda}\right)\!,   \eea}
\hspace{-1.4mm}and the first-order derivative with respect to $z$ is in (\ref{dev3}) where in this case the metric in (\ref{eta}) becomes 
\bea \eta=z_0^2+y^2+x^2.  \eea
Since $\eta$ is an even function with respect to $x$ and $y$, the cross-terms of different dimensions in the Fisher-information matrix are then zeros, which is diagonal with diagonal elements being
\bea \label{appA1} I_{ii}=\iint_{x^2+y^2\leq R^2} |\Delta s_i|^2\mathrm{d}x\mathrm{d}y. \eea
Calculating (\ref{appA1}) directly yields
{\setlength\arraycolsep{0pt}\bea  \label{appA2} I_{11}&=&I_{22}=\frac{z_0}{4\pi}\left(\frac{9}{4}g_1(7)+ \frac{4\pi^2}{\lambda^2}g_1(5)\right)\!,  \\
\label{appA3} I_{33}&=&\frac{z_0^3}{4\pi}\bigg(\frac{1}{4z_0^4}g_2(3) +\left(\frac{4\pi^2}{\lambda^2}- \frac{3}{2z_0^2}\right)g_2(5)+ \frac{9}{4}g_2(7)\!\bigg), \qquad \eea}
\hspace{-1.4mm}Utilizing the results in (\ref{g1n}) and (\ref{g2n}) and after some multiplications, the Fisher-information for different dimensions are then in (\ref{Ixy}) and (\ref{Iz}).

\bibliographystyle{IEEEtran}

\end{document}